\documentclass[showpacs,amsmath,amssymb,onecolumn,pra,reprint,superscriptaddress,nofootinbib]{revtex4-1}
\usepackage{amssymb}
\usepackage[dvips]{graphicx} % for figures
\usepackage{enumerate}
\usepackage{epsfig}
\usepackage{subfigure}
\usepackage{xcolor}
\usepackage[T1]{fontenc}
\usepackage[expert]{mathdesign}
\usepackage{fullpage}
\usepackage{amsthm,amsfonts,amssymb,amscd,mathrsfs,eufrak,xspace,framed}
\usepackage{amsmath}
\usepackage{color}
\usepackage{setspace}
\usepackage{url}
\usepackage{enumitem}
\usepackage{multirow}
\usepackage{xspace}
\usepackage{textcomp}
\newtheorem{theorem}{Theorem}

\newcommand{\bra}[1]{\mbox{$\left\langle #1 \right|$}}
\newcommand{\ket}[1]{\mbox{$\left| #1 \right\rangle$}}

\newcommand{\CWWfour}{CWW17$_4$\xspace}
\newcommand{\SSfour}{SS10$_4$\xspace}

\begin{document}
\title{Chau-Wang-Wong17 Scheme Is Experimentally More Feasible Than The Six-State Scheme}
\author{H. F. Chau}
\email{hfchau@hku.hk}
\affiliation{Department of Physics, University of Hong Kong, Pokfulam Road,
 Hong Kong}
\affiliation{Center of Theoretical and Computational Physics, University of
 Hong Kong, Pokfulam Road, Hong Kong}
\author{Zhen-Qiang Yin}
\email{yinzq@ustc.edu.cn}
\author{Shuang Wang}
\email{wshuang@ustc.edu.cn}
\author{Wei Chen}
\author{Zheng-Fu Han}
\affiliation{CAS Key Laboratory of Quantum Information, CAS Center For Excellence in Quantum Information and Quantum Physics,
	University of Science and Technology of China, Hefei 230026, China}
\affiliation{State Key Laboratory of Cryptology, P. O. Box 5159, Beijing 100878, P. R. China}
\date{\today}
\begin{abstract}
 Recently, Chau \emph{et al.} [Phys. Rev. A \textbf{95}, 022311 (2017)]
 reported a quantum-key-distribution (QKD) scheme using four-dimensional qudits.
 Surprisingly, as a function of the bit error rate of the raw key, the secret
 key rate of this scheme is equal to that of the (qubit-based) six-state
 scheme under one-way classical communication using ideal apparatus in the
 limit of arbitrarily long raw key length.
 Here we explain why this is the case in spite of the fact that these two
 schemes are not linearly related to each other.
 More importantly, we find that in terms of the four-dimensional dit error rate
 of the raw key, the Chau \emph{et al.}'s scheme can tolerate up to 21.6\%
 using one-way classical communications, which is better than the
 Sheridan and Scarani's scheme [Phys. Rev. A \text{82}, 030301(R) (2010)].
 In addition, we argue the experimental advantages of the Chau \emph{et al.}
 implementation over the standard six-state scheme and report a corresponding
 proof-of-principle experiment using passive basis selection with decoy states.
 We also compare our experiment with the recent high secret key rate
 implementation of the Sheridan and Scarani's scheme by Islam \emph{et al.}
 [Sci. Adv. \text{3}, e1701491].
\end{abstract}

\maketitle

\section{Introduction}
\label{Sec:intro}
 In theory, the six-state scheme~\cite{Bruss:6state:1998} is a powerful
 qubit-based quantum key distribution (QKD) scheme that tolerates higher
 channel noise (up to 12.6\% bit error rate (BER) using one-way classical
 communication using non-degenerate code~\cite{Lo:six:01}) than the BB84 protocol~\cite{Bennett:BB84:1984}.
 However, it does not attract much experimental attention because it requires
 more optical components, making it more lossy than the BB84
 protocol~\cite{Scarani:QKDrev:2008}.
 Along a different direction and based on an earlier work of
 Chau~\cite{Chau15}, his group reported an $N$-dimensional-qudit-based scheme
 that has exactly the same one-way secret key rate formula as a function of the
 bit error rate of the raw key for ideal equipment
 in the arbitrarily long raw key length limit provided that
 $N=4$~\cite{CWW2017}.
 We call their scheme the \CWWfour scheme, where the subscript refers to the
 Hilbert space dimension $N$ of the qudit used.

 It is instructive to find out the underlying reason for the agreement
 of the one-way key rate formulas of these two schemes.
 Here we first show that the six-state scheme cannot be imbedded in the
 \CWWfour scheme and vice versa.
 Then, we argue that this key rate agreement is out of a lucky coincidence.

 Along a different line, Sheridan and Scarani~\cite{SS10,SS10e}
 introduced a scheme using four-dimensional qudits as information carrier that
 we called the \SSfour scheme.
 In their scheme, states are prepared and measured either in the computational
 basis $\{ \ket{j} \}_{j=0}^3$ or its Fourier transformed basis $\{
 \sum_{k=0}^3 e^{\pi i j k/2} \ket{k} / 2 \}_{j=0}^3$.
 They showed that by using ideal apparatus and single photon source and in the
 limit of infinite raw key length, the \SSfour scheme tolerates up to an error
 rate of 18.9\% for qudit depolarizing channel~\cite{SS10,SS10e}.  Although not explicitly defined, it is
 evident from their analysis that they referred to the four-dimensional dit error rate\footnote{That is to say, there are four possible measurement outcomes for each qudit, say, $0,1,2,3$.  The dit error rate refers to the error rate of this sifted key expressed in dits.  One may convert this dit string to a bit string, say, by mapping $0,1,2,3$ to $00,01,10,11$, respectively.  And one may talk about the BER of this sifted bit string key.} (DER) of the raw key.
 And this translates to a tolerable BER of 12.6\%, which equals exactly
 that of the six-state scheme.
 Nevertheless, by carefully studying their proof, it is clear that if Alice and Bob just care about the average BER without looking into the three different four-dimensional DERs, their scheme can only tolerate up to 11.0\% BER just like the BB84 protocol because a channel that independently produces spin flip and phase shift errors to each qudit will produce the same BER as a qudit depolarizing channel.
 In fact, their proof implicitly showed that the secret key rate of (the unbiased basis selection version of) the \SSfour scheme is 1.5~times that of the six-state scheme partly because Alice and Bob has two rather than three bases to choose from.  We explicitly write down their argument in Appendix~\ref{Sec:SS10_rate}.

 It is instructive to compare the theoretical and experimental performances of
 the \CWWfour and the \SSfour schemes.
 Here we find that the \CWWfour scheme can tolerate up to a DER of 21.6\% in the raw key using one-way classical communciation in the infinite raw key length limit using ideal single photon source and detectors, which is higher than that of the \SSfour scheme.
 Based on earlier success of time-bin
 implementations~\cite{RRDPSexp1,RRDPSexp2,RRDPSexp3,RRDPSexp4,PassiveRRDPS,Chau15exp}
 of the round-robin differential-phase-shift~\cite{sasaki2014practical},
 the Chau15~\cite{Chau15} and the \SSfour~\cite{qkdtimebinqudits} schemes, it makes sense to implement the
 \CWWfour scheme via the time-bin representation.
 So, we perform such an experiment using passive basis selection and decoy
 states, and discuss its practical advantages over the original six-state scheme as well as comparing it with the recent implementation of a biased basis selection version of the \SSfour scheme by Islam \emph{et al.}~\cite{qkdtimebinqudits}.

\section{Differences between the six-state and the \CWWfour schemes}
\label{Sec:six-state_CWW4}
 Recall that in the six-state scheme, Alice and Bob prepare and measure qubits
 in one of the basis states of the following three mutually unbiased bases
 (MUBs) $\{\ket{0}, \ket{1}\}$, $\{(\ket{0}\pm\ket{1})/\sqrt{2}\}$, and
 $\{(\ket{0}\pm i\ket{1})/\sqrt{2} \}$~\cite{Bruss:6state:1998}; whereas those
 states for the \CWWfour scheme using four-dimensional qudits are $\ket{\psi_{jk}^\pm} \equiv (\ket{j}\pm
 \ket{k})/\sqrt{2}$ for $0 \leqslant j < k \leqslant 3$~\cite{CWW2017}.
 Here ${\mathcal B}_\ell = \{ \ket{j}\colon 0\leqslant j < \ell \}$ is an
 orthonormal basis of the $\ell$-dimensional Hilbert space.
 (Note that Ref.~\cite{CWW2017} labelled the four basis states using finite
 field notation to emphasize its underlying mathematical structure.  Here we just label them from $0$ to $3$ for the convenience of experimentalists.)
 In the six-state scheme, the raw key bit of Alice (Bob) is assigned to $0$ if
 the preparation (measurement) state is $\ket{0}$, $(\ket{0}+\ket{1})/\sqrt{2}$
 or $(\ket{0}+i\ket{1})/\sqrt{2}$.  Otherwise, it is assigned as
 $1$~\cite{Bruss:6state:1998}.
 For the \CWWfour scheme, the 12~preparation and measurement states form
 three set of MUBs in the four-dimensional Hilbert space.  Therefore, each
 prepared or measured qudit corresponds to two raw bits.  For instance, for
 states prepared or measured in the basis $\{
 (\ket{j}+(-1)^k\ket{j+2})/\sqrt{2} \colon j,k=0,1 \}$, the raw bits are $j$ and $k$~\cite{CWW2017}.

 For generation of the final secret key from the raw bits, we follow the standard Shor and Preskill procedure~\cite{Shor:Preskill:2000} adapted to the decoy state situation~\cite{Wang:Decoy:2005,Lo:Decoy:2005,MXF:2way:2006}.  using one-way classical communication.  And for simplicity, we use the so-called random key assignment in the sense that whenever Bob does not detect a signal, he will randomly and uniformly assign the ``measurement result'' as one of the four pairs of bits $00, 01, 10$ and $11$.  Finally, in case more than one of the Bob's detectors click, we randomly assign Bob's measurement result~\cite{Fung:2011:Squash}.

 By considering a corresponding entanglement-distillation protocol of four-dimensional qudits, Chau \emph{et al.}~\cite{CWW2017} proved that for ideal apparatus and in the infinitely long sifted key length limit, the one-way key rate of the \CWWfour scheme equals to that of the six-state scheme.  But they do not know why.  In Appendix~\ref{Sec:proof}, we show that although the six-state and the \CWWfour schemes
 have the same one-way key rate as a function of the BER in the raw key in the case of ideal source and detectors,
 they are very different schemes in the sense that one cannot be imbedded
 in the other.
 In simple terms, it means that the preparation methods of these two schemes
 are so different that they are not linearly related.
 Consequently, we believe that they have the same one-way key rate formulas in
 the ideal apparatus situation because of a lucky coincidence.
 And this coincidence comes from the following observation.
 As Alice and Bob each randomly picks one of the three MUBs, they in effect
 completely mix the phase and spin flip errors in the quantum
 channel~\cite{CWW2017}.
 For the \CWWfour scheme, the worst-case one-way key rate for a given raw key
 rate happens when the phase and spin flip errors are
 independent~\cite{CWW2017}.
 Thus, this worst-case key rate can be computed as if each raw bit had passed
 through a depolarizing channel --- the very same situation of the six-state
 scheme.

\begin{figure*}[t]
 \centering{
	\includegraphics[width=7.5cm]{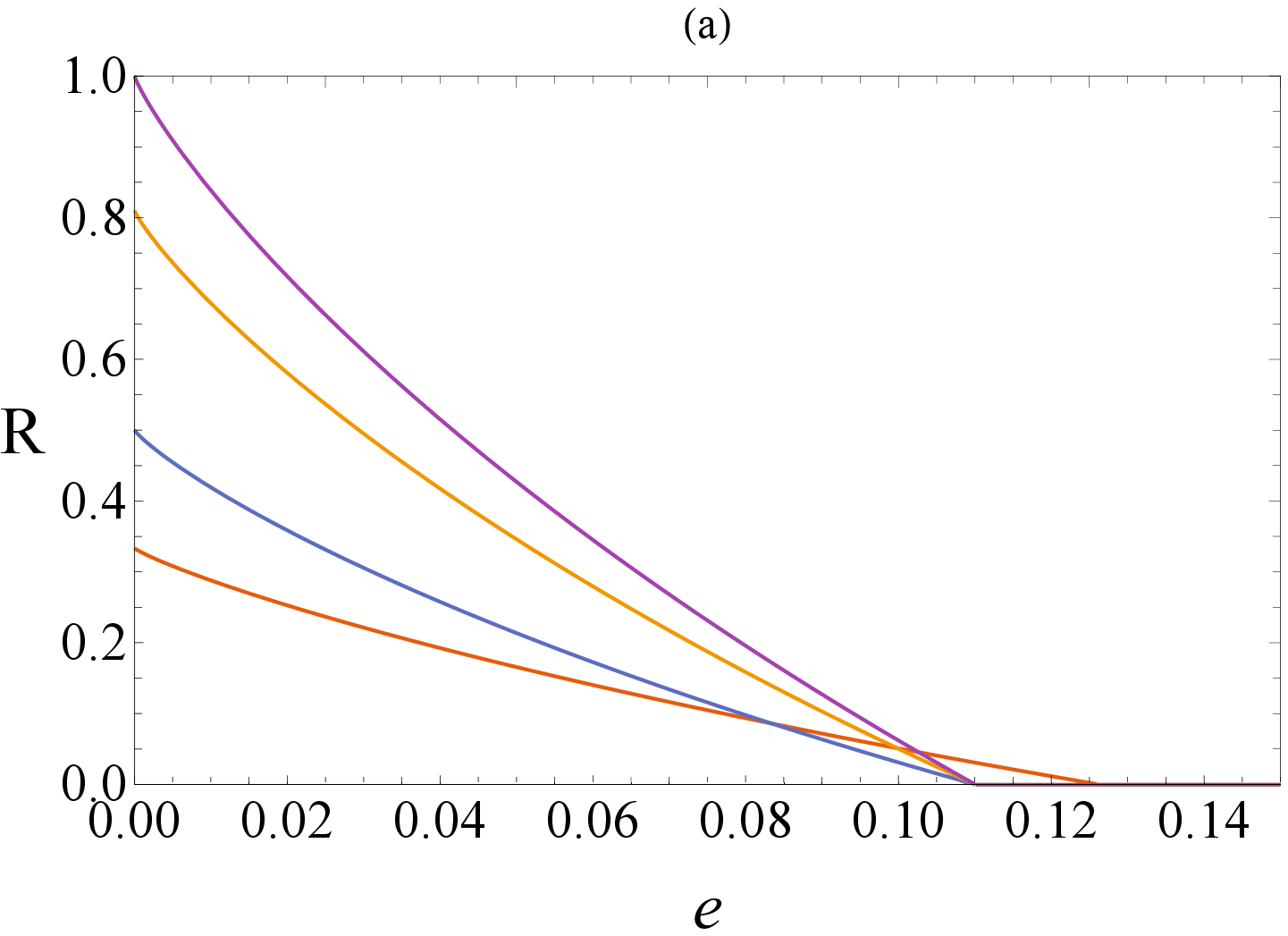}
	\qquad
	\includegraphics[width=7.5cm]{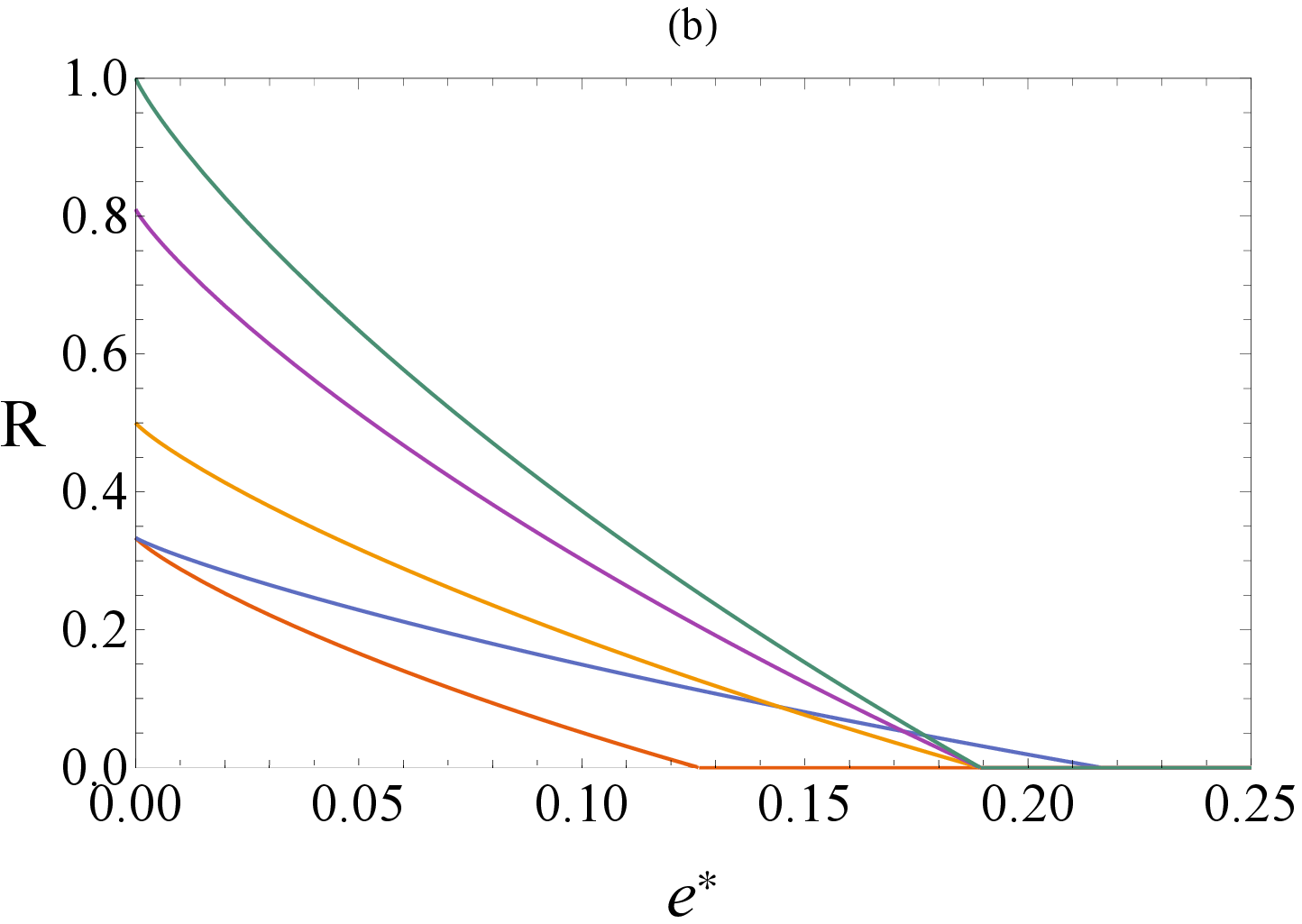}
	}
	\caption{The secret key rates $R$ of various protocols for ideal apparatus as a function of (a) the BER $e$ and (b) the DER $e^*$ in the raw key.
	In~(a), the curves from the top to the bottom are the extremely biased version, the version used in Islam \emph{et al.}~\cite{qkdtimebinqudits} and the unbiased version of the \SSfour scheme, and the six-state scheme (and hence also the \CWWfour scheme).
	In~(b), the curves from the top to the bottom are the extremely biased version, the version used in Islam \emph{et al.}~\cite{qkdtimebinqudits} and the unbiased version of the \SSfour scheme, the \CWWfour scheme and the six-state scheme.
	\label{F:key_rates}}
\end{figure*}

\section{One-way secret key rate formulas for the original and the modified \CWWfour schemes}
\label{Sec:key_rates}
 The one-way secret key rate formula for the \CWWfour scheme for ideal single
 photon source in Ref.~\cite{CWW2017} can extended to the case of using
 standard decoy state via the flagging method first explicitly presented in
 Ref.~\cite{MXF:2way:2006}.  (See Ref.~\cite{Chaunew17} for detail.)
 It is given by
\begin{align}
 R &= \frac{q Q_\mu}{s} \left\{ -H_2(\{E_\mu^g\}_{g=0}^3)
  \vphantom{\sum_{g=0}^3} \right. \nonumber \\
 & \qquad \left. +\min \sum_{g=0}^3
 \Omega^g [s-H_2(\{\delta_p^g\}_p)] \right\} ,
 \label{E:decoy-key-rate}
\end{align}
 where $s=2$ is the conversion factor from a $4$-dimensional dit to
 2~bits~\cite{Chau15,CWW2017}, $q=1/3$ is the chance that Alice and Bob use the
 same basis, $Q_\mu = \sum_{n=0}^{+\infty} Y_n
 \mu^n \exp(-\mu) / n!$ is the overall gain of the signal, $Y_n$ is the
 conditional probability that Bob's detector(s) clicks given Alice emits $n$
 photons.
 Also, $H_2(\{x_k\}_{k=1}^N) = H_2(x_1,\dots,x_N) \equiv -\sum_{k=1}^N x_k
 \log_2 x_k$ is the binary entropy function provided that $\sum_{k=1}^N x_k = 1$, $E_\mu^g = \sum_{n=0}^{+\infty}
 e_n^g Y_n \mu^n \exp(-\mu) / ( Q_\mu n!)$ is the overall rate of the
 $4$-dimensional signal dit with error $g$, and $e_n^g$ is the corresponding
 conditional probability given Alice emits $n$ photons~\cite{MXF:2way:2006}.
 Here, $g(=0,1,2,3)$ error means that the bitwise addition modulo 2 of the
 two corresponding raw secret bits of Alice and Bob is
 $g$~\cite{Chau15,CWW2017}.
 (Using the example in the last sentence of the second last paragraph, the
 least significant bit of $g$ is $j$ and the most significant bit of $g$ is
 $k$.)
 Further, $\Omega^g = Y_1 \mu \exp(-\mu) e^g_1 / Q_\mu$ is the fraction of
 single photon that experience error $g$, and $\delta_p^g$ is the ``phase
 error rate'' of those dits in the raw key with error $g$.
 Finally, the minimization is over all $\delta_p^g$'s that are consistent with
 the deduced values of $e_1^g$'s.
 Specifically, from Ref.~\cite{Chaunew17} which extends the work of
 Ref.~\cite{CWW2017}, we have $e_1^0 = A+B+C+D$, $e_1^1 = 2(B+D)$,
 $e_1^2 = 2(C+D)$ and $e_1^3 = 4D$ with $A+3B+3C+9D=1$ and $0\leqslant A,B,C,D
 \leqslant 1$.
 Moreover, the BER of the single photon event in the raw key is
 $e = (e_1^1+e_1^2)/2+e_1^3$.
 Then, the minimization (subject to a fixed $e$) is attained when $H_2(\{\delta_p^3\}_p) =
 H_2(\{D/4D,D/4D,D/4D,D/4D\}) = 2$, $H_2(\{\delta_p^1\}_p) =
 H_2(\{B/[2(B+D)],B/[2(B+D)], D/[2(B+D)],D/[2(B+D)]\})$,
 $H_2(\{\delta_p^2\}_p) =
 H_2(\{C/[2(C+D)],C/[2(C+D)], D/[2(C+D)],D/[2(C+D)]\})$
 and $H_2(\{\delta_p^0\}_p) = H_2(\{A/(A+B+C+D),B/(A+B+C+D),C/(A+B+C+D),
 D/(A+B+C+D)\})$.

 Let us also compute the secret key rate as a function of the DER of the raw key.
 By concavity of the entropy function, the worst case secret key rate occurs when $E_\mu^1 = E_\mu^2 =
E_\mu^3$ and $e_1^1 = e_1^2 = e_1^3$ in the infinite raw key length limit.
Moreover, using the notations in Ref.~\cite{CWW2017} (with the minor changes of
using labels $0$ to $3$ instead of labels in finite field notations),
$e_{01} = e_{02} = e_{03} = e_{10} = e_{20} = e_{30} = e_{11} = e_{23} = e_{32}$, $e_{12} = e_{21} = e_{33}$ and $e_{13} = e_{22} = e_{31}$.
Thus,
$\{ \delta_p^0 \} = \{ e_{00}/e_1^0, e_{01}/e_1^0, e_{01}/e_1^0, e_{01}/e_1^0\}$ and
$\{ \delta_p^g \} = \{ e_{01}/e_1^g, e_{01}/e_1^g, e_{12}/e_1^g, e_{13}/e_1^g\}$ for all $g\ne 0$.
Combined with the constraints that $e_{00} + 9 e_{01} + 3 e_{12} + 3 e_{13} = 1$
and the DER of the raw key $e^* = 6 e_{01} + 3 e_{12} + 3 e_{13}$, we may minimize the second term in Eq.~\eqref{E:decoy-key-rate} to get the secret key rate.  (See Appendix~\ref{Sec:modified_CWW4_proof} for detail.)

Fig.~\ref{F:key_rates} depicts the (one-way) secret key rates of various QKD protocols for ideal photon source and detector without the need of decoy.
It shows that the modified \CWWfour scheme tolerates up to 14.4\% BER (or 21.6\% DER), which is better than the six-state and the \SSfour schemes.  Note that the modified \CWWfour scheme is better than the unbiased (extremely biased) version of the \SSfour scheme when the DER $e^*$ exceeds 14.4\% (17.7\%).  This demonstrates the advantage of the modified \CWWfour scheme over the \SSfour scheme for very noisy channel.
Finally, we remark that the curves in Fig.~\ref{F:key_rates} show the worst case secret key rates.  In practice, one should use Eq.~\eqref{E:decoy-key-rate} to compute the secret key rate $R$ because it takes the error rates $e_1^g$ for all $g$ into account.  In a lot of cases, it gives a better value of $R$ that the worst case situation.

\begin{figure*}[t]
	\includegraphics[width=14cm]{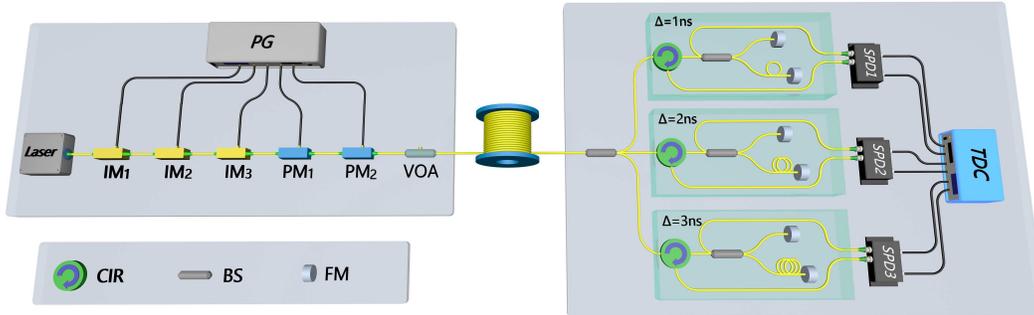}
	\caption{Experimental setup. IM: intensity modulator; PM: phase modulator; VOA: variable optical attenuator; PG: pattern generator; CIR: circulator; BS: beam splitter; FM: Faraday mirror; SPD: single photon detector; TDC: time-to-digital convertor.}
	\label{setup}
\end{figure*}

\section{Experimental results and their comparison with other QKD schemes}
\label{Sec:expt}
 We now report to our \CWWfour scheme experiment using time-bin representation.
Our implementation is shown in Fig.~\ref{setup}. At Alice's site, a pulse train with a repetition rate of $1\ $GHz is generated by modulating a  continuous wave laser using the first LiNbO$_3$ intensity modulator (IM$_1$), and only two random pulses indexed by $j$ and $k$ ($j,k\in\{0,1,2,3\}$ and $j<k$) in each packets of 4~time slots ($4\ $ns) are allowed to pass IM$_2$. IM$_3$ is employed to implement the decoy states method \cite{Hwang:Decoy:2003,Wang:Decoy:2005,Lo:Decoy:2005}, by which each packet is randomly modulated into signal state (whose intensity is $\mu$ photons per packet), and two decoy states (whose intensities are $\nu$ and $\upsilon$ photons per packet respectively.). Naturally, information leakage in single photon emission is decided according to the security proof, and then decoy-state method is straight-forward provided the phase randomized weak coherent source is employed. Then, the first phase modulator (PM$_1$) adds phase $-\pi/2$ or $\pi/2$ on each pulse to encode the key bits, PM$_2$ randomize the global phase of each packet. Finally, a large attenuation is added on these encoded pulses by a variable optical attenuator (VOA). Alice's output quantum state of each packet is $\ket{\psi_{jk}^\pm}$.

%need this package \usepackage{multirow}
\begin{table*}[t]
	\centering
	\caption {The length of fiber ($l$), the mean photon numbers ($\mu$, $\nu$ and $\upsilon$) per packet and yield ($Q$) per packet, error rates ($E^1$, $E^{2}$, and $E^{3}$), and the secret key rate per packet (R). By increasing optical misalignment intentionally, additional observations of $E^2$ and the corresponding secret key rate are listed in brackets.}\label{results}
	\renewcommand\arraystretch{1.25}%adjut linespace
	\setlength{\tabcolsep}{9pt}
%	\begin{ruledtabular}
	\begin{tabular}{cccccccc}
	\hline
	\hline
		$l$ &\multicolumn{2}{c}{Intensity} & $Q$ & $E^{1}$ & $E^{2}$ & $E^{3}$ & $R$\\
		\hline
		\multirow{3}{*}{$50\ $km} & $\mu$ & $0.66$ & $5.63\times10^{-3}$ & $0.216\%$ & $1.81\%$ $(15.1\%)$ & $0.217\%$ & \multirow{3}{*}{$7.31\times 10^{-4}$ $(1.64\times 10^{-5})$}\\
		& $\nu$ & $0.04$ & $3.56\times10^{-4}$ & $1.24\%$ & $2.77\%$ $(19.4\%)$ & $1.24\%$\\
		& $\upsilon$ & $0.0016$ & $2.92\times10^{-5}$ & $13.4\%$ & $14.2\%$ $(20.4\%)$ & $13.4\%$ \\
	\hline
	\hline
	\end{tabular}
%	\end{ruledtabular}
\end{table*}

At Bob's site, the passive scheme based on a $1 \times 3$ beam splitter (BS) is used to implement a high-speed, stable, and low-loss decoding measurement. Following the passive measurement-delay choice,  three unbalanced Faraday-Michelson interferometers (FMI) with $\Delta=k-j\in\{1,2,3\}$ temporal delays are employed to make the $j$th pulse interfere with $k$th pulse. One three-port optical circulator (CIR), one $50:50$ BS and two Faraday mirrors (FM) constitute a FMI, whose two output ports are connected to two channels of one single photon detector (SPD), respectively. There are totally three double-channel SPDs, and all detection events are recorded by a time-to-digital convertor (TDC) that records the time-tagged and which-channel information.
(In principle, we should add a narrow-bandwidth filter in Bob's side to prevent wavelength-dependent beam-splitter attack~\cite{PhysRevA.84.062308}.   But as we are in effect uisng a monochromatic laser source, we decided not to do so in this demonstrative experiment to simply matter.)

Compared to the active scheme, the passive approach of the variable-delay interferometer with three delay values can be characterized as follows: (i) Highspeed, the passive choice among three $1\ $ns, $2\ $ns and $3\ $ns delay FMIs has no speed limits in principle \cite{RRDPSexp2}; (ii) Stable, three FMIs are insensitive to polarization variations, and also independent, so we can actively and independently compensate the phase shift of each interferometer, which is placed in small and separate ABS plastic case and on heating plate to keep its temperature a little above the environment temperature. Each interferometer is individually stabilized by controlling the current of corresponding heating plate, the feedback signals include counts of SPD when only IM$_1$ works and error rate during key distribution procedure; (iii) Low-loss, the insertion loss (IL) of each FMI depends on ILs of the CIR, BS, and FM, so the ILs of these three FMIs are almost identical, and approximate to be $0.80$~dB. Nevertheless, we have to point out that the passive approach needs more SPDs, and also cannot implement bias basis choices as easily as the active scheme.

Albeit the present implementation shares the same key rate formula with the six-state QKD protocol, they have its own features experimentally. The main advantage of the present scheme lies in its less demand for phase encoding. In fact, a main drawback of a time-phase coding six-state system is that Alice must module her phase modulator with four phases $0$, $\pi/2$, $\pi$ and $3 \pi/2$, but our scheme only needs $-\pi/2$ and $\pi/2$ phases. Two-phase modulation and $V_{\pi}$ peak-to-peak voltage facilitate the realization of high speed QKD, since four-phase modulation is more complex than two-phase modulation and $V_{3 \pi/2}$ peak-to-peak voltage is higher than $V_{\pi}$. In a word, \CWWfour has a simple phase-coding device, thus is particularly significant for practical QKD networks.

For each SPD, both channels are based on InGaAs/InP avalanche photodiodes (APD), and operated in gated Geiger mode with sine-wave filtering method \cite{lpf}. The two output ports of each FMI are different, the output port from the 3rd port of CIR includes about additional $0.40\ $dB IL compared to the other one, which is connected to the channel of the corresponding SPD with lower detection efficiency. In order to achieve optimal performances and low error rate, we first add a width discriminator to remove the wider filtered avalanche signals in each channel of SPD \cite{he2017sine}, and then set the measurement time window to $800\ $ps in TDC for all SPDs. The average detection efficiency, dark count rate, and after-pulse probability of these three double-channel SPDs is approximately $20.23\%$, $2.58\times10^{-6}$ per gate, and $1.05\%$, respectively. Here, the IL of CIR from the 2nd port to the 3rd port and the reduction effect by setting measurement time window have been included in the detection efficiency, the dark count rate of one SPD is the sum value of both channels.

The experimental results using standard telecom fiber channels of length $l = 50\ $km are listed in Table~\ref{results}.
The universal squash model~\cite{Fung:2011:Squash} is applicable to our experimental setup; and we used the data processing procedure there to handle events with multiple detector click.  Actually, these events contribute to no more than 0.003\% of the raw key, which has negligible effect on the secret key rate.
The BER and DER of the raw key are $(0.216\% + 1.81\%)/2 + 0.217\% = 1.23\%$ and $0.216\%+1.81\%+0.217\% = 2.24\%$, respectively.
With the data given in Table~\ref{results}, we calculate the corresponding parameters for single photon emission through standard decoy states formulas \cite{Limetal2014}, which are
$Y_1=8.38\times 10^{-3}$, $e^{1}_1=e^{3}_1=0.21\%$ and $e^{2}_1=1.9\%$.
According to Eq.~\eqref{E:decoy-key-rate}, the secret key rate is $R=7.31\times 10^{-4}$ per packet for the original \CWWfour scheme.

To verify the high error tolerance of the proposed protocol, we intentionally lower the optical interference to result in high $E^2$, which are listed in the brackets in Table~\ref{results}. We find that positive secret key rate is achieved even when $E^2$ equals $15.1\%$ (which means 7.88\% BER and 15.5\% DER).
%As a comparison, let's consider a BB84 experiment with the same experimental conditions of ours. Then, it is observed that error rate $E_{BB84}=(E^1+E^2)/2+E^3$ and same counting rates $Q_\mu$, $Q_{\nu}$ and $Q_{\upsilon}$ as ours. According to decoy states formulas \cite{Limetal2014}, one obtains that $Y_1=8.38\times 10^{-3}$ and $e_1=10.6\%$. These parameters of a BB84 system cannot yield positive secret key rate $R=-Q_{\mu}H_2(E_{\mu})+\mu e^{-\mu}Y_1(1-H_2(e_1))$. This example clearly shows the advantage of the \CWWfour scheme over BB84.
As a comparison, let us consider a BB84 experiment with the same experimental conditions of ours.
More precisely, we consider a BB84 experiment in time-bin representation in phase encoding and that we intentionally lower the optical interference through the same optical mis-alignment.
Clearly, $E_\text{BB84}$ in this setup is $E^2$.  Besides, using the decoy state formulas in Ref.~\cite{Limetal2014}, we arrive at
$Y_1=8.38\times 10^{-3}$, $e_1=20.5\%$.
Hence, it is impossible for the BB84 scheme to generate a secure key using this setup. This example, therefore, shows the advantage of the \CWWfour scheme over BB84.

To compare our experiment with the recent experiment by Islam \emph{et al.}~\cite{qkdtimebinqudits}, we remark that their experiment aimed at producing the highest possible secret key rate (measured in unit of secret bit per second rather than per packet).  That is why they applied a strong bias of $90\%$ in choosing the computational (that is, the time-bin) basis.  Besides, they sent photons about 10~times faster than we do; and they optimized the decoy intensities and probabilites of using different decoys (though they did not show these probabilities explicitly, making readers hard to verify their computed secret key rates).
Moreover, in order to detect photons in the computational and the Fourier transformed bases, they used 8~photon detectors whereas we only use 3.
Using similar commercial fiber with $l = 50$~km, their observed DER (for single photon events) is at least $3.73\%$, which is higher than our $2.24\%$.  This suggests that the \CWWfour scheme may also have an edge in actual experimental setup in terms of noise control.

\section{Summary}
\label{Sec:summary}
In summary, we show that the \CWWfour scheme is not linearly related to the
six-state scheme although they share the same one-way key rate.
Although these two schemes may be related to each other in some subtle ways,
say, due to some hidden symmetry, we believe that it is simply a coincidence
that they share the same one-way key rate.
We also perform experiments to demonstrate the ease of implementing the
\CWWfour scheme in time-bin representation over the six-state scheme.
Moreover, we show the noise resiliency of the \CWWfour scheme over the BB84
scheme by artificially increasing the channel noise to such a level that no
secure key can be distilled out from the latter scheme while a secret key
can be generated from the former one though at a very low rate.

Last but not least, one may consider the following reduced \CWWfour scheme.
Instead of using a total of 12~states --- four states from each of the three
basis, Alice and Bob may use the following three pairs of states each chosen
from a basis:
$(|0\rangle\pm|1\rangle)/\sqrt{2}$, $(|0\rangle\pm|2\rangle)/\sqrt{2}$ and
$(|0\rangle\pm|3\rangle)/\sqrt{2}$.
That is to say, Alice randomly prepares these states and Bob randomly measures
them in one of the three bases.
Instead of two bits, they get one bit per transmitted four-dimensional
qudit that is prepared and measured in the same basis.
Furthermore, they reject the qudit if the measured state does not belong to
the two-dimensional Hilbert subspace used in the state preparation.
Using the argument in Ref.~\cite{CWW2017}, it is not difficult to show that the
one-way key rate of this modified scheme is less than or equal to half that of
the six-state scheme, with equality holds if none of the qudits is rejected.
(The factor of half in the key rate comes from the fact that Alice and Bob
use four-dimensional qudits instead of qubits in transmitting quantum
information in the channel.)
By considering all other possibilities, we see that out of the possible
reduced \CWWfour schemes that uses a pair of quantum state per basis, the above
one already gives the highest one-way key rate.
In this regard, we conclude that the \CWWfour scheme, which uses
four-dimensional qudits in 12~states that associate with three different bases
is the minimalist round-robin differential-phase-shift-type of scheme that
achieves the one-way key rate of the six-state scheme with the least number of
states and Hilbert space dimension.
This analysis shows the advantage of obtaining more than one bit of raw key
per channel use. However, in practice, the performance of the proposed protocol may be not so strong,
since there may be more noises introduced in the detection process \cite{practicalRRDPS}.

\appendix

\section{Secret key rate of the \SSfour scheme with unbiased basis selection for ideal apparatus and photon source in the infinite raw key length limit}
\label{Sec:SS10_rate}
One way to see this is that just like the BB84 scheme, the key rate of the \SSfour scheme equals $R = (2-H_\text{spin}-H_\text{phase})/(2\times 2)$ where $H_\text{spin}$ and $H_\text{phase}$ are the entropies of the spin flip and phase errors of the raw key, respectively.  In addition, the first and the second $2$ in the denominator are due to the conversion of a four-dimensional dit to a bit and the use of two bases with equal probabilities, respectively.
For the unbiased protocol using both computational and Fourier transformed states to generate the sifted key, $H_\text{spin} = H_\text{phase}$ in the worst case scenario.  Therefore, $R = (1-H_\text{spin})/2$.
When using the DER as the figure of merit, the worst case situation occurs when $H_\text{spin} = H_2(1-3e/2,e/2,e/2,e/2)$ where $e$ is the BER of the raw key (and hence the DER of the raw key equals $3e/2$).
Consequently, the secret key rate is $1.5$ times that of the six-state scheme.
In contrast, if the BER is used as the figure of merit, the worst case situation occurs when $H_\text{spin} = 2 H_2(1-e,e)$.  This gives the BB84 key rate.

\section{Proof of the unitary non-equivalence between the six-state scheme and the \CWWfour scheme}
\label{Sec:proof}
\begin{theorem}
 The six-state and the \CWWfour schemes are not linearly related in the sense
 that for any $m,n>0$, there is no monomorphism imbedding every
 normalized preparation state of the $m$ qubits for the six-state scheme to
 those of the $n$ qudits for the \CWWfour scheme and vice versa.
\end{theorem}

\begin{proof}
 The following three observations are needed.
 First, no complex numbers $a,b$ satisfy $|a|=|b|=|a\pm b|/\sqrt{2}=|a\pm i b|
 \sqrt{2}$.
 Second, all normalized $n$ qudit preparation states in the \CWWfour scheme can
 be written in the form $\sum_{j_1,\ldots,j_n=0}^3 \alpha_{j_1,\ldots,j_n}
 \ket{j_1,\ldots,j_n}$ with exactly $2^n$ non-zero $\alpha_{j_1,\ldots,j_n}$'s
 each with the same magnitude of $2^{-n/2}$.
 Third, all normalized $m$ qubit preparation states in the six-state scheme can
 be written in the form $\sum_{k_1,\ldots,k_m=0}^1 \beta_{k_1,\ldots,k_m}
 \ket{k_1,\ldots,k_m}$ with $2^\ell$ non-zero $\beta_{k_1,\ldots,k_m}$'s, each
 with the same magnitude of $2^{-\ell/2}$ for some $\ell\geqslant 0$.

 Now suppose the contrary that there is an injective linear operator $L$
 sending normalized preparation states of the $m$ qubit six-state scheme to the
 normalized $n$ qudit \CWWfour scheme.
 We consider only the case of $m=1$ here as the general case can be proven in
 the same way.  Applying the first two observations to $L\ket{\varphi}$, where
 $\ket{\varphi}$ is one of the preparation states in the six-state scheme, we
 conclude that $\bra{j_1,\ldots,j_n}L\ket{0} = 0$ whenever
 $\bra{j_1,\ldots,j_n}L\ket{1}\ne 0$.  However, this means
 $L(\ket{0}+\ket{1})/\sqrt{2}$ has to be a sum of $2^{n+1}$ rather than $2^n$
 non-zero terms over the basis ${\mathcal B}_4^{\otimes n}$, which contradicts
 the second observation.

 Suppose there is an injective linear operator $L'$ sending normalized
 preparation states of the $n$ qudit \CWWfour scheme to the normalized $m$
 qubit six-state scheme.
 By the Hilbert space dimension consideration, we know that $m\geqslant 2n$.
 Again, we consider only the case of $n=1$ for the
 general case can be proven in the same way.
 We may assume that $L'\ket{\psi_{01}^+} = \ket{0}^{\otimes m}$.
 (Otherwise, we simply amend a unitary transformation on the $m$ qubits after
 $L'$ to make it so.)
 Consider $L'\ket{\psi_{12}^-} = \ket{0}^{\otimes m} - L'\ket{\psi_{02}^+}$
 and hence $\bra{0,\ldots,0}L'\ket{\psi_{12}^-} +
 \bra{0,\ldots,0}L'\ket{\psi_{02}^+} = 1$.
 From observation three, this is possible only if $L'\ket{\psi_{02}^+}$ can be
 expressed as a sum of two non-zero terms over the basis
 ${\mathcal B}_2^{\otimes m}$, say, $\ket{0}^{\otimes (m-1)}\otimes (\ket{0} +
 \ket{1})/\sqrt{2}$.
 The same argument plus the injectivity of $L'$ lead to $L'\ket{\psi_{03}^+} =
 \ket{0}^{\otimes (m-2)}\otimes (\ket{0}+\ket{1})\otimes\ket{0}/\sqrt{2}$, say.
 However, this implies $L'\ket{\psi_{23}^-} = \ket{0}^{\otimes (m-2)} \otimes
 (\ket{0}\otimes\ket{1}-\ket{1}\otimes\ket{0})/\sqrt{2}$, which is not a
 preparation state of the $m$ qubit six-state scheme.
 Hence, $L'$ does not exist; and this completes our proof.
\end{proof}

 We remark on passing that using the same idea in the above proof, there is
 no monomorphism imbedding every normalized preparation state of the $m$ qudits
 for the \SSfour scheme to those of the $n$ qudits for the \CWWfour scheme.

\section{Detailed derivation of the secret key rate formula for the \CWWfour scheme in the infinite raw key length limit as a function of the DER}
\label{Sec:modified_CWW4_proof}
We need to minimize the second term in Eq.~\eqref{E:decoy-key-rate}, which can be rewritten as
\begin{align}
 & \min \sum_{g=0}^3 \Omega^g \left[ s - H_2(\{\delta_p^g\}_p) \right]
  \nonumber \\
 ={} & \Omega \left[ s - \max \sum_{g=0}^3 e_1^g H_2(\{\delta_p^g\}_p) \right]
  \nonumber \\
 ={} & \Omega \left\{ s - \max \left[ H_2(\{e_{jk}\}_{j,k=0}^3) - H_2(\{e_1^g\}_{g=0}^3) \right] \right\} \nonumber \\
 ={} & \Omega \left[ s + H_2(\{e_1^g\}_{g=0}^3) - \max H_2(\{e_{jk}\}_{j,k=0}^3) \right] ,
 \label{E:appendixB1}
\end{align}
where $\Omega = \sum_{g=0}^3 \Omega^g = Y_1\mu \exp(-\mu)/Q_\mu$.
Clearly, the maximum in the above equation is attained when $e_{12}=e_{13}$.
Combined with the sum rule $e_{00} + 9 e_{01} + 3 e_{12} + 3 e_{13} = 1$ and the expression for the DER $e^* = 6e_{01} + 3e_{12} + 3e_{13}$, the $H_2(\{e_{jk}\}_{j,k=0}^3)$ becomes a function of $e^*$, $e_{01}$ only.
By considering $d H_2(\{e_{jk}\}_{j,k=0}^3) / d e_{01}$ and
$d^2 H_2(\{e_{jk}\}_{j,k=0}^3) / d e_{01}^2$, it is easy to see that
$H_2(\{e_{jk}\}_{j,k=0}^3)$ is maximized when
\begin{equation}
 f(e_{01}) \equiv e_{01}^3 - (1 - 3e_{01} - e^*) (e^* - 6 e_{01})^2  = 0 \label{E:appendixB2}
\end{equation}
 in the domain $[0,e^*/6]$.  As $f(0) < 0$, $f(e^*/6) > 0$ and $df/d e_{01} > 0$
 for $e_{01}\in [0,e^*/6]$, there is an unique root for $f$ in the domain
 $[0,e^*/6]$.
 Substituting this root back to $H_2$ maximizes it; and hence we obtain the
 secret key rate through Eq.~\eqref{E:decoy-key-rate}.

\medskip
\acknowledgments
H.F.C. is supported by the Research Grant Council of the HKSAR Government (Grant No. 17304716). Other authors are supported by the National Key Research And Development Program of China (Grant No.2016YFA0302600), the National Natural Science Foundation of China (Grant Nos. 61822115, 61775207, 61622506, 61627820, 61575183), and the Anhui Initiative in Quantum Information Technologies.

\bibliographystyle{apsrev4-1}

\bibliography{Biblisource}

\begin{thebibliography}{28}%
\makeatletter
\providecommand \@ifxundefined [1]{%
 \@ifx{#1\undefined}
}%
\providecommand \@ifnum [1]{%
 \ifnum #1\expandafter \@firstoftwo
 \else \expandafter \@secondoftwo
 \fi
}%
\providecommand \@ifx [1]{%
 \ifx #1\expandafter \@firstoftwo
 \else \expandafter \@secondoftwo
 \fi
}%
\providecommand \natexlab [1]{#1}%
\providecommand \enquote  [1]{``#1''}%
\providecommand \bibnamefont  [1]{#1}%
\providecommand \bibfnamefont [1]{#1}%
\providecommand \citenamefont [1]{#1}%
\providecommand \href@noop [0]{\@secondoftwo}%
\providecommand \href [0]{\begingroup \@sanitize@url \@href}%
\providecommand \@href[1]{\@@startlink{#1}\@@href}%
\providecommand \@@href[1]{\endgroup#1\@@endlink}%
\providecommand \@sanitize@url [0]{\catcode `\\12\catcode `\$12\catcode
  `\&12\catcode `\#12\catcode `\^12\catcode `\_12\catcode `\%12\relax}%
\providecommand \@@startlink[1]{}%
\providecommand \@@endlink[0]{}%
\providecommand \url  [0]{\begingroup\@sanitize@url \@url }%
\providecommand \@url [1]{\endgroup\@href {#1}{\urlprefix }}%
\providecommand \urlprefix  [0]{URL }%
\providecommand \Eprint [0]{\href }%
\providecommand \doibase [0]{http://dx.doi.org/}%
\providecommand \selectlanguage [0]{\@gobble}%
\providecommand \bibinfo  [0]{\@secondoftwo}%
\providecommand \bibfield  [0]{\@secondoftwo}%
\providecommand \translation [1]{[#1]}%
\providecommand \BibitemOpen [0]{}%
\providecommand \bibitemStop [0]{}%
\providecommand \bibitemNoStop [0]{.\EOS\space}%
\providecommand \EOS [0]{\spacefactor3000\relax}%
\providecommand \BibitemShut  [1]{\csname bibitem#1\endcsname}%
\let\auto@bib@innerbib\@empty
%</preamble>
\bibitem [{\citenamefont {Bru\ss{}}(1998)}]{Bruss:6state:1998}%
  \BibitemOpen
  \bibfield  {author} {\bibinfo {author} {\bibfnamefont {D.}~\bibnamefont
  {Bru\ss{}}},\ }\href {\doibase 10.1103/PhysRevLett.81.3018} {\bibfield
  {journal} {\bibinfo  {journal} {Phys. Rev. Lett.}\ }\textbf {\bibinfo
  {volume} {81}},\ \bibinfo {pages} {3018} (\bibinfo {year}
  {1998})}\BibitemShut {NoStop}%
\bibitem [{\citenamefont {Lo}(2001)}]{Lo:six:01}%
  \BibitemOpen
  \bibfield  {author} {\bibinfo {author} {\bibfnamefont {H.-K.}\ \bibnamefont
  {Lo}},\ }\href@noop {} {\bibfield  {journal} {\bibinfo  {journal} {Quantum
  Inf.~Comput.}\ }\textbf {\bibinfo {volume} {1}},\ \bibinfo {pages} {81}
  (\bibinfo {year} {2001})}\BibitemShut {NoStop}%
\bibitem [{\citenamefont {Bennett}\ and\ \citenamefont
  {Brassard}(1984)}]{Bennett:BB84:1984}%
  \BibitemOpen
  \bibfield  {author} {\bibinfo {author} {\bibfnamefont {C.~H.}\ \bibnamefont
  {Bennett}}\ and\ \bibinfo {author} {\bibfnamefont {G.}~\bibnamefont
  {Brassard}},\ }in\ \href@noop {} {\emph {\bibinfo {booktitle} {Proceedings of
  the IEEE International Conference on Computers, Systems and Signal
  Processing}}}\ (\bibinfo  {publisher} {IEEE Press},\ \bibinfo {address} {New
  York},\ \bibinfo {year} {1984})\ pp.\ \bibinfo {pages} {175--179}\BibitemShut
  {NoStop}%
\bibitem [{\citenamefont {Scarani}\ \emph {et~al.}(2009)\citenamefont
  {Scarani}, \citenamefont {Bechmann-Pasquinucci}, \citenamefont {Cerf},
  \citenamefont {Du\ifmmode~\check{s}\else \v{s}\fi{}ek}, \citenamefont
  {L\"utkenhaus},\ and\ \citenamefont {Peev}}]{Scarani:QKDrev:2008}%
  \BibitemOpen
  \bibfield  {author} {\bibinfo {author} {\bibfnamefont {V.}~\bibnamefont
  {Scarani}}, \bibinfo {author} {\bibfnamefont {H.}~\bibnamefont
  {Bechmann-Pasquinucci}}, \bibinfo {author} {\bibfnamefont {N.~J.}\
  \bibnamefont {Cerf}}, \bibinfo {author} {\bibfnamefont {M.}~\bibnamefont
  {Du\ifmmode~\check{s}\else \v{s}\fi{}ek}}, \bibinfo {author} {\bibfnamefont
  {N.}~\bibnamefont {L\"utkenhaus}}, \ and\ \bibinfo {author} {\bibfnamefont
  {M.}~\bibnamefont {Peev}},\ }\href {\doibase 10.1103/RevModPhys.81.1301}
  {\bibfield  {journal} {\bibinfo  {journal} {Rev. Mod. Phys.}\ }\textbf
  {\bibinfo {volume} {81}},\ \bibinfo {pages} {1301} (\bibinfo {year}
  {2009})}\BibitemShut {NoStop}%
\bibitem [{\citenamefont {Chau}(2015)}]{Chau15}%
  \BibitemOpen
  \bibfield  {author} {\bibinfo {author} {\bibfnamefont {H.~F.}\ \bibnamefont
  {Chau}},\ }\href {\doibase 10.1103/PhysRevA.92.062324} {\bibfield  {journal}
  {\bibinfo  {journal} {Phys. Rev. A}\ }\textbf {\bibinfo {volume} {92}},\
  \bibinfo {pages} {062324} (\bibinfo {year} {2015})}\BibitemShut {NoStop}%
\bibitem [{\citenamefont {Chau}\ \emph
  {et~al.}(2017{\natexlab{a}})\citenamefont {Chau}, \citenamefont {Wang},\ and\
  \citenamefont {Wong}}]{CWW2017}%
  \BibitemOpen
  \bibfield  {author} {\bibinfo {author} {\bibfnamefont {H.~F.}\ \bibnamefont
  {Chau}}, \bibinfo {author} {\bibfnamefont {Q.}~\bibnamefont {Wang}}, \ and\
  \bibinfo {author} {\bibfnamefont {C.}~\bibnamefont {Wong}},\ }\href {\doibase
  10.1103/PhysRevA.95.022311} {\bibfield  {journal} {\bibinfo  {journal} {Phys.
  Rev. A}\ }\textbf {\bibinfo {volume} {95}},\ \bibinfo {pages} {022311}
  (\bibinfo {year} {2017}{\natexlab{a}})}\BibitemShut {NoStop}%
\bibitem [{\citenamefont {Sheridan}\ and\ \citenamefont
  {Scarani}(2010)}]{SS10}%
  \BibitemOpen
  \bibfield  {author} {\bibinfo {author} {\bibfnamefont {L.}~\bibnamefont
  {Sheridan}}\ and\ \bibinfo {author} {\bibfnamefont {V.}~\bibnamefont
  {Scarani}},\ }\href@noop {} {\bibfield  {journal} {\bibinfo  {journal} {Phys.
  Rev. A}\ }\textbf {\bibinfo {volume} {82}},\ \bibinfo {pages} {030301(R)}
  (\bibinfo {year} {2010})}\BibitemShut {NoStop}%
\bibitem [{\citenamefont {Sheridan}\ and\ \citenamefont
  {Scarani}(2011)}]{SS10e}%
  \BibitemOpen
  \bibfield  {author} {\bibinfo {author} {\bibfnamefont {L.}~\bibnamefont
  {Sheridan}}\ and\ \bibinfo {author} {\bibfnamefont {V.}~\bibnamefont
  {Scarani}},\ }\href@noop {} {\bibfield  {journal} {\bibinfo  {journal} {Phys.
  Rev. A}\ }\textbf {\bibinfo {volume} {83}},\ \bibinfo {pages} {039901(E)}
  (\bibinfo {year} {2011})}\BibitemShut {NoStop}%
\bibitem [{\citenamefont {Takesue}\ \emph {et~al.}(2015)\citenamefont
  {Takesue}, \citenamefont {Sasaki}, \citenamefont {Tamaki},\ and\
  \citenamefont {Koashi}}]{RRDPSexp1}%
  \BibitemOpen
  \bibfield  {author} {\bibinfo {author} {\bibfnamefont {H.}~\bibnamefont
  {Takesue}}, \bibinfo {author} {\bibfnamefont {T.}~\bibnamefont {Sasaki}},
  \bibinfo {author} {\bibfnamefont {K.}~\bibnamefont {Tamaki}}, \ and\ \bibinfo
  {author} {\bibfnamefont {M.}~\bibnamefont {Koashi}},\ }\href
  {http://dx.doi.org/10.1038/nphoton.2015.173} {\bibfield  {journal} {\bibinfo
  {journal} {Nat Photon}\ }\textbf {\bibinfo {volume} {9}},\ \bibinfo {pages}
  {827} (\bibinfo {year} {2015})}\BibitemShut {NoStop}%
\bibitem [{\citenamefont {{Wang}}\ \emph {et~al.}(2015)\citenamefont {{Wang}},
  \citenamefont {{Yin}}, \citenamefont {{Chen}}, \citenamefont {{He}},
  \citenamefont {{Song}}, \citenamefont {{Li}}, \citenamefont {{Zhang}},
  \citenamefont {{Zhou}}, \citenamefont {{Guo}},\ and\ \citenamefont
  {{Han}}}]{RRDPSexp2}%
  \BibitemOpen
  \bibfield  {author} {\bibinfo {author} {\bibfnamefont {S.}~\bibnamefont
  {{Wang}}}, \bibinfo {author} {\bibfnamefont {Z.-Q.}\ \bibnamefont {{Yin}}},
  \bibinfo {author} {\bibfnamefont {W.}~\bibnamefont {{Chen}}}, \bibinfo
  {author} {\bibfnamefont {D.-Y.}\ \bibnamefont {{He}}}, \bibinfo {author}
  {\bibfnamefont {X.-T.}\ \bibnamefont {{Song}}}, \bibinfo {author}
  {\bibfnamefont {H.-W.}\ \bibnamefont {{Li}}}, \bibinfo {author}
  {\bibfnamefont {L.-J.}\ \bibnamefont {{Zhang}}}, \bibinfo {author}
  {\bibfnamefont {Z.}~\bibnamefont {{Zhou}}}, \bibinfo {author} {\bibfnamefont
  {G.-C.}\ \bibnamefont {{Guo}}}, \ and\ \bibinfo {author} {\bibfnamefont
  {Z.-F.}\ \bibnamefont {{Han}}},\ }\href
  {http://dx.doi.org/10.1038/nphoton.2015.209} {\bibfield  {journal} {\bibinfo
  {journal} {Nat Photon}\ }\textbf {\bibinfo {volume} {9}},\ \bibinfo {pages}
  {832} (\bibinfo {year} {2015})}\BibitemShut {NoStop}%
\bibitem [{\citenamefont {Li}\ \emph {et~al.}(2016)\citenamefont {Li},
  \citenamefont {Cao}, \citenamefont {Dai}, \citenamefont {Lin}, \citenamefont
  {Zhang}, \citenamefont {Chen}, \citenamefont {Xu}, \citenamefont {Guan},
  \citenamefont {Liao}, \citenamefont {Yin}, \citenamefont {Zhang},
  \citenamefont {Ma}, \citenamefont {Peng},\ and\ \citenamefont
  {Pan}}]{RRDPSexp3}%
  \BibitemOpen
  \bibfield  {author} {\bibinfo {author} {\bibfnamefont {Y.-H.}\ \bibnamefont
  {Li}}, \bibinfo {author} {\bibfnamefont {Y.}~\bibnamefont {Cao}}, \bibinfo
  {author} {\bibfnamefont {H.}~\bibnamefont {Dai}}, \bibinfo {author}
  {\bibfnamefont {J.}~\bibnamefont {Lin}}, \bibinfo {author} {\bibfnamefont
  {Z.}~\bibnamefont {Zhang}}, \bibinfo {author} {\bibfnamefont
  {W.}~\bibnamefont {Chen}}, \bibinfo {author} {\bibfnamefont {Y.}~\bibnamefont
  {Xu}}, \bibinfo {author} {\bibfnamefont {J.-Y.}\ \bibnamefont {Guan}},
  \bibinfo {author} {\bibfnamefont {S.-K.}\ \bibnamefont {Liao}}, \bibinfo
  {author} {\bibfnamefont {J.}~\bibnamefont {Yin}}, \bibinfo {author}
  {\bibfnamefont {Q.}~\bibnamefont {Zhang}}, \bibinfo {author} {\bibfnamefont
  {X.}~\bibnamefont {Ma}}, \bibinfo {author} {\bibfnamefont {C.-Z.}\
  \bibnamefont {Peng}}, \ and\ \bibinfo {author} {\bibfnamefont {J.-W.}\
  \bibnamefont {Pan}},\ }\href {\doibase 10.1103/PhysRevA.93.030302} {\bibfield
   {journal} {\bibinfo  {journal} {Phys. Rev. A}\ }\textbf {\bibinfo {volume}
  {93}},\ \bibinfo {pages} {030302} (\bibinfo {year} {2016})}\BibitemShut
  {NoStop}%
\bibitem [{\citenamefont {{Yin}}\ \emph {et~al.}(2018)\citenamefont {{Yin}},
  \citenamefont {{Wang}}, \citenamefont {{Chen}}, \citenamefont {{Han}},
  \citenamefont {{Wang}}, \citenamefont {{Guo}},\ and\ \citenamefont
  {{Han}}}]{RRDPSexp4}%
  \BibitemOpen
  \bibfield  {author} {\bibinfo {author} {\bibfnamefont {Z.-Q.}\ \bibnamefont
  {{Yin}}}, \bibinfo {author} {\bibfnamefont {S.}~\bibnamefont {{Wang}}},
  \bibinfo {author} {\bibfnamefont {W.}~\bibnamefont {{Chen}}}, \bibinfo
  {author} {\bibfnamefont {Y.-G.}\ \bibnamefont {{Han}}}, \bibinfo {author}
  {\bibfnamefont {R.}~\bibnamefont {{Wang}}}, \bibinfo {author} {\bibfnamefont
  {G.-C.}\ \bibnamefont {{Guo}}}, \ and\ \bibinfo {author} {\bibfnamefont
  {Z.-F.}\ \bibnamefont {{Han}}},\ }\href
  {http://dx.doi.org/10.1038/s41467-017-02211-x} {\bibfield  {journal}
  {\bibinfo  {journal} {Nat Commun}\ }\textbf {\bibinfo {volume} {9}},\
  \bibinfo {pages} {457} (\bibinfo {year} {2018})}\BibitemShut {NoStop}%
\bibitem [{\citenamefont {Guan}\ \emph {et~al.}(2015)\citenamefont {Guan},
  \citenamefont {Cao}, \citenamefont {Liu}, \citenamefont {Shen-Tu},
  \citenamefont {Pelc}, \citenamefont {Fejer}, \citenamefont {Peng},
  \citenamefont {Ma}, \citenamefont {Zhang},\ and\ \citenamefont
  {Pan}}]{PassiveRRDPS}%
  \BibitemOpen
  \bibfield  {author} {\bibinfo {author} {\bibfnamefont {J.-Y.}\ \bibnamefont
  {Guan}}, \bibinfo {author} {\bibfnamefont {Z.}~\bibnamefont {Cao}}, \bibinfo
  {author} {\bibfnamefont {Y.}~\bibnamefont {Liu}}, \bibinfo {author}
  {\bibfnamefont {G.-L.}\ \bibnamefont {Shen-Tu}}, \bibinfo {author}
  {\bibfnamefont {J.~S.}\ \bibnamefont {Pelc}}, \bibinfo {author}
  {\bibfnamefont {M.~M.}\ \bibnamefont {Fejer}}, \bibinfo {author}
  {\bibfnamefont {C.-Z.}\ \bibnamefont {Peng}}, \bibinfo {author}
  {\bibfnamefont {X.}~\bibnamefont {Ma}}, \bibinfo {author} {\bibfnamefont
  {Q.}~\bibnamefont {Zhang}}, \ and\ \bibinfo {author} {\bibfnamefont {J.-W.}\
  \bibnamefont {Pan}},\ }\href {\doibase 10.1103/PhysRevLett.114.180502}
  {\bibfield  {journal} {\bibinfo  {journal} {Phys. Rev. Lett.}\ }\textbf
  {\bibinfo {volume} {114}},\ \bibinfo {pages} {180502} (\bibinfo {year}
  {2015})}\BibitemShut {NoStop}%
\bibitem [{\citenamefont {Wang}\ \emph {et~al.}(2018)\citenamefont {Wang},
  \citenamefont {Yin}, \citenamefont {Chau}, \citenamefont {Chen},
  \citenamefont {Wang}, \citenamefont {Guo},\ and\ \citenamefont
  {Han}}]{Chau15exp}%
  \BibitemOpen
  \bibfield  {author} {\bibinfo {author} {\bibfnamefont {S.}~\bibnamefont
  {Wang}}, \bibinfo {author} {\bibfnamefont {Z.-Q.}\ \bibnamefont {Yin}},
  \bibinfo {author} {\bibfnamefont {H.~F.}\ \bibnamefont {Chau}}, \bibinfo
  {author} {\bibfnamefont {W.}~\bibnamefont {Chen}}, \bibinfo {author}
  {\bibfnamefont {C.}~\bibnamefont {Wang}}, \bibinfo {author} {\bibfnamefont
  {G.-C.}\ \bibnamefont {Guo}}, \ and\ \bibinfo {author} {\bibfnamefont
  {Z.-F.}\ \bibnamefont {Han}},\ }\href
  {http://stacks.iop.org/2058-9565/3/i=2/a=025006} {\bibfield  {journal}
  {\bibinfo  {journal} {Quantum Science and Technology}\ }\textbf {\bibinfo
  {volume} {3}},\ \bibinfo {pages} {025006} (\bibinfo {year}
  {2018})}\BibitemShut {NoStop}%
\bibitem [{\citenamefont {Sasaki}\ \emph {et~al.}(2014)\citenamefont {Sasaki},
  \citenamefont {Yamamoto},\ and\ \citenamefont
  {Koashi}}]{sasaki2014practical}%
  \BibitemOpen
  \bibfield  {author} {\bibinfo {author} {\bibfnamefont {T.}~\bibnamefont
  {Sasaki}}, \bibinfo {author} {\bibfnamefont {Y.}~\bibnamefont {Yamamoto}}, \
  and\ \bibinfo {author} {\bibfnamefont {M.}~\bibnamefont {Koashi}},\
  }\href@noop {} {\bibfield  {journal} {\bibinfo  {journal} {Nature}\ }\textbf
  {\bibinfo {volume} {509}},\ \bibinfo {pages} {475} (\bibinfo {year}
  {2014})}\BibitemShut {NoStop}%
\bibitem [{\citenamefont {Islam}\ \emph {et~al.}(2017)\citenamefont {Islam},
  \citenamefont {Lim}, \citenamefont {Cahall}, \citenamefont {Kim},\ and\
  \citenamefont {Gauthier}}]{qkdtimebinqudits}%
  \BibitemOpen
  \bibfield  {author} {\bibinfo {author} {\bibfnamefont {N.~T.}\ \bibnamefont
  {Islam}}, \bibinfo {author} {\bibfnamefont {C.~C.~W.}\ \bibnamefont {Lim}},
  \bibinfo {author} {\bibfnamefont {C.}~\bibnamefont {Cahall}}, \bibinfo
  {author} {\bibfnamefont {J.}~\bibnamefont {Kim}}, \ and\ \bibinfo {author}
  {\bibfnamefont {D.~J.}\ \bibnamefont {Gauthier}},\ }\href {\doibase
  10.1126/sciadv.1701491} {\bibfield  {journal} {\bibinfo  {journal} {Science
  Advances}\ }\textbf {\bibinfo {volume} {3}},\ \bibinfo {pages} {e1701491}
  (\bibinfo {year} {2017})}\BibitemShut {NoStop}%
\bibitem [{\citenamefont {Shor}\ and\ \citenamefont
  {Preskill}(2000)}]{Shor:Preskill:2000}%
  \BibitemOpen
  \bibfield  {author} {\bibinfo {author} {\bibfnamefont {P.~W.}\ \bibnamefont
  {Shor}}\ and\ \bibinfo {author} {\bibfnamefont {J.}~\bibnamefont
  {Preskill}},\ }\href@noop {} {\bibfield  {journal} {\bibinfo  {journal}
  {Phys.~Rev.~Lett.~}\ }\textbf {\bibinfo {volume} {85}},\ \bibinfo {pages}
  {441} (\bibinfo {year} {2000})}\BibitemShut {NoStop}%
\bibitem [{\citenamefont {Wang}(2005)}]{Wang:Decoy:2005}%
  \BibitemOpen
  \bibfield  {author} {\bibinfo {author} {\bibfnamefont {X.-B.}\ \bibnamefont
  {Wang}},\ }\href@noop {} {\bibfield  {journal} {\bibinfo  {journal}
  {Phys.~Rev.~Lett.~}\ }\textbf {\bibinfo {volume} {94}},\ \bibinfo {pages}
  {230503} (\bibinfo {year} {2005})}\BibitemShut {NoStop}%
\bibitem [{\citenamefont {Lo}\ \emph {et~al.}(2005)\citenamefont {Lo},
  \citenamefont {Ma},\ and\ \citenamefont {Chen}}]{Lo:Decoy:2005}%
  \BibitemOpen
  \bibfield  {author} {\bibinfo {author} {\bibfnamefont {H.-K.}\ \bibnamefont
  {Lo}}, \bibinfo {author} {\bibfnamefont {X.}~\bibnamefont {Ma}}, \ and\
  \bibinfo {author} {\bibfnamefont {K.}~\bibnamefont {Chen}},\ }\href@noop {}
  {\bibfield  {journal} {\bibinfo  {journal} {Phys.~Rev.~Lett.~}\ }\textbf
  {\bibinfo {volume} {94}},\ \bibinfo {pages} {230504} (\bibinfo {year}
  {2005})}\BibitemShut {NoStop}%
\bibitem [{\citenamefont {Ma}\ \emph {et~al.}(2006)\citenamefont {Ma},
  \citenamefont {Fung}, \citenamefont {Dupuis}, \citenamefont {Chen},
  \citenamefont {Tamaki},\ and\ \citenamefont {Lo}}]{MXF:2way:2006}%
  \BibitemOpen
  \bibfield  {author} {\bibinfo {author} {\bibfnamefont {X.}~\bibnamefont
  {Ma}}, \bibinfo {author} {\bibfnamefont {C.-H.~F.}\ \bibnamefont {Fung}},
  \bibinfo {author} {\bibfnamefont {F.}~\bibnamefont {Dupuis}}, \bibinfo
  {author} {\bibfnamefont {K.}~\bibnamefont {Chen}}, \bibinfo {author}
  {\bibfnamefont {K.}~\bibnamefont {Tamaki}}, \ and\ \bibinfo {author}
  {\bibfnamefont {H.-K.}\ \bibnamefont {Lo}},\ }\href@noop {} {\bibfield
  {journal} {\bibinfo  {journal} {Phys.~Rev.~A}\ }\textbf {\bibinfo {volume}
  {74}},\ \bibinfo {pages} {032330} (\bibinfo {year} {2006})}\BibitemShut
  {NoStop}%
\bibitem [{\citenamefont {Fung}\ \emph {et~al.}(2011)\citenamefont {Fung},
  \citenamefont {Chau},\ and\ \citenamefont {Lo}}]{Fung:2011:Squash}%
  \BibitemOpen
  \bibfield  {author} {\bibinfo {author} {\bibfnamefont {C.-H.~F.}\
  \bibnamefont {Fung}}, \bibinfo {author} {\bibfnamefont {H.~F.}\ \bibnamefont
  {Chau}}, \ and\ \bibinfo {author} {\bibfnamefont {H.-K.}\ \bibnamefont
  {Lo}},\ }\href {\doibase 10.1103/PhysRevA.84.020303} {\bibfield  {journal}
  {\bibinfo  {journal} {Phys. Rev. A}\ }\textbf {\bibinfo {volume} {84}},\
  \bibinfo {pages} {020303} (\bibinfo {year} {2011})}\BibitemShut {NoStop}%
\bibitem [{\citenamefont {Chau}\ \emph
  {et~al.}(2017{\natexlab{b}})\citenamefont {Chau}, \citenamefont {Wong},
  \citenamefont {Huang},\ and\ \citenamefont {Wang}}]{Chaunew17}%
  \BibitemOpen
  \bibfield  {author} {\bibinfo {author} {\bibfnamefont {H.~F.}\ \bibnamefont
  {Chau}}, \bibinfo {author} {\bibfnamefont {C.}~\bibnamefont {Wong}}, \bibinfo
  {author} {\bibfnamefont {T.}~\bibnamefont {Huang}}, \ and\ \bibinfo {author}
  {\bibfnamefont {Q.}~\bibnamefont {Wang}},\ }\href@noop {} {\enquote {\bibinfo
  {title} {Provably secure key rate analysis of finite-key-length qudit-based
  decoy state quantum key distributions},}\ } (\bibinfo {year}
  {2017}{\natexlab{b}}),\ \bibinfo {note} {in preparation}\BibitemShut
  {NoStop}%
\bibitem [{\citenamefont {Hwang}(2003)}]{Hwang:Decoy:2003}%
  \BibitemOpen
  \bibfield  {author} {\bibinfo {author} {\bibfnamefont {W.-Y.}\ \bibnamefont
  {Hwang}},\ }\href {\doibase 10.1103/PhysRevLett.91.057901} {\bibfield
  {journal} {\bibinfo  {journal} {Phys. Rev. Lett.}\ }\textbf {\bibinfo
  {volume} {91}},\ \bibinfo {pages} {057901} (\bibinfo {year}
  {2003})}\BibitemShut {NoStop}%
\bibitem [{\citenamefont {Li}\ \emph {et~al.}(2011)\citenamefont {Li},
  \citenamefont {Wang}, \citenamefont {Huang}, \citenamefont {Chen},
  \citenamefont {Yin}, \citenamefont {Li}, \citenamefont {Zhou}, \citenamefont
  {Liu}, \citenamefont {Zhang}, \citenamefont {Guo}, \citenamefont {Bao},\ and\
  \citenamefont {Han}}]{PhysRevA.84.062308}%
  \BibitemOpen
  \bibfield  {author} {\bibinfo {author} {\bibfnamefont {H.-W.}\ \bibnamefont
  {Li}}, \bibinfo {author} {\bibfnamefont {S.}~\bibnamefont {Wang}}, \bibinfo
  {author} {\bibfnamefont {J.-Z.}\ \bibnamefont {Huang}}, \bibinfo {author}
  {\bibfnamefont {W.}~\bibnamefont {Chen}}, \bibinfo {author} {\bibfnamefont
  {Z.-Q.}\ \bibnamefont {Yin}}, \bibinfo {author} {\bibfnamefont {F.-Y.}\
  \bibnamefont {Li}}, \bibinfo {author} {\bibfnamefont {Z.}~\bibnamefont
  {Zhou}}, \bibinfo {author} {\bibfnamefont {D.}~\bibnamefont {Liu}}, \bibinfo
  {author} {\bibfnamefont {Y.}~\bibnamefont {Zhang}}, \bibinfo {author}
  {\bibfnamefont {G.-C.}\ \bibnamefont {Guo}}, \bibinfo {author} {\bibfnamefont
  {W.-S.}\ \bibnamefont {Bao}}, \ and\ \bibinfo {author} {\bibfnamefont
  {Z.-F.}\ \bibnamefont {Han}},\ }\href {\doibase 10.1103/PhysRevA.84.062308}
  {\bibfield  {journal} {\bibinfo  {journal} {Phys. Rev. A}\ }\textbf {\bibinfo
  {volume} {84}},\ \bibinfo {pages} {062308} (\bibinfo {year}
  {2011})}\BibitemShut {NoStop}%
\bibitem [{\citenamefont {Walenta}\ \emph {et~al.}(2012)\citenamefont
  {Walenta}, \citenamefont {Lunghi}, \citenamefont {Guinnard}, \citenamefont
  {Houlmann}, \citenamefont {Zbinden},\ and\ \citenamefont {Gisin}}]{lpf}%
  \BibitemOpen
  \bibfield  {author} {\bibinfo {author} {\bibfnamefont {N.}~\bibnamefont
  {Walenta}}, \bibinfo {author} {\bibfnamefont {T.}~\bibnamefont {Lunghi}},
  \bibinfo {author} {\bibfnamefont {O.}~\bibnamefont {Guinnard}}, \bibinfo
  {author} {\bibfnamefont {R.}~\bibnamefont {Houlmann}}, \bibinfo {author}
  {\bibfnamefont {H.}~\bibnamefont {Zbinden}}, \ and\ \bibinfo {author}
  {\bibfnamefont {N.}~\bibnamefont {Gisin}},\ }\href {\doibase
  10.1063/1.4749802} {\bibfield  {journal} {\bibinfo  {journal} {Journal of
  Applied Physics}\ }\textbf {\bibinfo {volume} {112}},\ \bibinfo {pages}
  {063106} (\bibinfo {year} {2012})},\ \Eprint
  {http://arxiv.org/abs/http://dx.doi.org/10.1063/1.4749802}
  {http://dx.doi.org/10.1063/1.4749802} \BibitemShut {NoStop}%
\bibitem [{\citenamefont {He}\ \emph {et~al.}(2017)\citenamefont {He},
  \citenamefont {Wang}, \citenamefont {Chen}, \citenamefont {Yin},
  \citenamefont {Qian}, \citenamefont {Zhou}, \citenamefont {Guo},\ and\
  \citenamefont {Han}}]{he2017sine}%
  \BibitemOpen
  \bibfield  {author} {\bibinfo {author} {\bibfnamefont {D.-Y.}\ \bibnamefont
  {He}}, \bibinfo {author} {\bibfnamefont {S.}~\bibnamefont {Wang}}, \bibinfo
  {author} {\bibfnamefont {W.}~\bibnamefont {Chen}}, \bibinfo {author}
  {\bibfnamefont {Z.-Q.}\ \bibnamefont {Yin}}, \bibinfo {author} {\bibfnamefont
  {Y.-J.}\ \bibnamefont {Qian}}, \bibinfo {author} {\bibfnamefont
  {Z.}~\bibnamefont {Zhou}}, \bibinfo {author} {\bibfnamefont {G.-C.}\
  \bibnamefont {Guo}}, \ and\ \bibinfo {author} {\bibfnamefont {Z.-F.}\
  \bibnamefont {Han}},\ }\href@noop {} {\bibfield  {journal} {\bibinfo
  {journal} {Applied Physics Letters}\ }\textbf {\bibinfo {volume} {110}},\
  \bibinfo {pages} {111104} (\bibinfo {year} {2017})}\BibitemShut {NoStop}%
\bibitem [{\citenamefont {Lim}\ \emph {et~al.}(2014)\citenamefont {Lim},
  \citenamefont {Curty}, \citenamefont {Walenta}, \citenamefont {Xu},\ and\
  \citenamefont {Zbinden}}]{Limetal2014}%
  \BibitemOpen
  \bibfield  {author} {\bibinfo {author} {\bibfnamefont {C.~C.~W.}\
  \bibnamefont {Lim}}, \bibinfo {author} {\bibfnamefont {M.}~\bibnamefont
  {Curty}}, \bibinfo {author} {\bibfnamefont {N.}~\bibnamefont {Walenta}},
  \bibinfo {author} {\bibfnamefont {F.}~\bibnamefont {Xu}}, \ and\ \bibinfo
  {author} {\bibfnamefont {H.}~\bibnamefont {Zbinden}},\ }\href@noop {}
  {\bibfield  {journal} {\bibinfo  {journal} {Phys. Rev. A}\ }\textbf {\bibinfo
  {volume} {89}},\ \bibinfo {pages} {022307} (\bibinfo {year}
  {2014})}\BibitemShut {NoStop}%
\bibitem [{\citenamefont {Zhang}\ \emph {et~al.}(2017)\citenamefont {Zhang},
  \citenamefont {Yuan}, \citenamefont {Cao},\ and\ \citenamefont
  {Ma}}]{practicalRRDPS}%
  \BibitemOpen
  \bibfield  {author} {\bibinfo {author} {\bibfnamefont {Z.}~\bibnamefont
  {Zhang}}, \bibinfo {author} {\bibfnamefont {X.}~\bibnamefont {Yuan}},
  \bibinfo {author} {\bibfnamefont {Z.}~\bibnamefont {Cao}}, \ and\ \bibinfo
  {author} {\bibfnamefont {X.}~\bibnamefont {Ma}},\ }\href
  {http://stacks.iop.org/1367-2630/19/i=3/a=033013} {\bibfield  {journal}
  {\bibinfo  {journal} {New Journal of Physics}\ }\textbf {\bibinfo {volume}
  {19}},\ \bibinfo {pages} {033013} (\bibinfo {year} {2017})}\BibitemShut
  {NoStop}%
\end{thebibliography}%

\end{document}